\let\arxiv1

\if\arxiv1 
    \pdfoutput=1 
\fi

\documentclass[a4paper,UKenglish,cleveref,autoref,thm-restate]{lipics-v2021}

\usepackage[utf8]{inputenc}

\if\arxiv1
    \hideLIPIcs  
\fi

\title{GKAT with Hoare Hypotheses}

\author{Jurriaan Rot}{Radboud University, Nijmegen, The Netherlands}{jurriaan.rot@ru.nl}{https://orcid.org/0000-0002-1404-6232}{}

\author{Todd Schmid}{Bucknell University, Lewisburg, USA}{t.schmid@bucknell.edu}{https://orcid.org/0000-0002-9838-2363}{}

\author{Jana Wagemaker}{Radboud University, Nijmegen, The Netherlands}{jana.wagemaker@ru.nl}{https://orcid.org/0000-0002-8616-3905}{}

\authorrunning{J. Rot, T. Schmid, and J. Wagemaker}

\Copyright{Jurriaan Rot, Todd Schmid, and Jana Wagemaker} 



\relatedversion{} 

\funding{This research was partially supported by the Dutch research council (NWO) under grant no.\ VI.Veni.242.134 (VerHyp).}

\nolinenumbers 

\EventEditors{Ana Sokolova and Patrick Totzke}
\EventNoEds{2}
\EventLongTitle{37th International Conference on Concurrency Theory (CONCUR 2026)}
\EventShortTitle{CONCUR 2026}
\EventAcronym{CONCUR}
\EventYear{2026}
\EventDate{September 1--4, 2026}
\EventLocation{Liverpool, UK}
\EventLogo{}
\SeriesVolume{391}
\ArticleNo{7}

\usepackage[dvipsnames]{xcolor}
\usepackage{amsmath,amsthm,amssymb}
\usepackage{
    mathpartir, 
    enumerate,
	thmtools
}
\usepackage{hyperref}
\hypersetup{
    colorlinks=true,
    linkcolor=black,
    filecolor=blue,      
    urlcolor=blue,
    citecolor=ForestGreen
    }

\newcommand{\BExp}{\mathit{BExp}}
\newcommand{\GExp}{\mathit{GExp}}
\newcommand{\At}{\mathit{At}}
\newcommand{\GL}{\mathit{GL}}

\newcommand{\BA}{\mathsf{BA}}

\renewcommand{\L}{{\mathcal L}}

\newcommand{\purple}[1]{\mathcolor{RedViolet}{#1}}
\newcommand{\blue}[1]{\mathcolor{Blue}{#1}}

\newcommand{\IF}[3]{\mathtt{\purple{if}}~{\TEST{#1}}~\mathtt{\purple{then}}~{#2}~\mathtt{\purple{else}}~{#3}}
\newcommand{\WHILE}[2]{
    \mathtt{\purple{while}}~{\TEST{#1}}
    ~\mathtt{\purple{do}}~{#2}
}
\newcommand{\ASSERT}[1]{\mathtt{\purple{assert}}~\TEST{#1}}

\newcommand{\TEST}[1]{\blue{#1}}

\newcommand{\GKAT}{\mathsf{GKAT}}
\newcommand{\transl}[1]{T_{#1}}

\newcommand{\tr}[1]{
    \raisebox{-1pt}{
        \(\xrightarrow{#1}\)
    }
}

\newcommand{\enq}{\mathsf{enq}}

\newcommand{\len}{\mathsf{len}}

\newcommand{\insp}{\mathsf{insp}}
\newcommand{\NN}{\mathbb{N}}
\newcommand{\Queue}{\mathcal{Q}}

\newcommand{\bigO}{\mathcal{O}}

\allowdisplaybreaks

\begin{document}

\maketitle

\begin{abstract}
    Guarded Kleene Algebra with Tests (GKAT) is a variant of Kleene algebra which allows for reasoning about simple imperative programs,
	and which features a decision procedure for program equivalence in nearly linear time. 
	In the current paper, we address the challenge of reasoning under assumptions about these programs.
	In particular, we develop a form of \emph{Hoare hypotheses}, which allow modelling basic domain knowledge on pre- and postconditions of uninterpreted basic programs, and which are well-developed for classical Kleene algebra but not yet for GKAT.
	We show that the resulting axiomatisation is sound and complete. 
	We then extend Hoare hypotheses to the more general form of \emph{word hypotheses}.
	Based on an automata-theoretic approach, we show that equivalence of GKAT under word hypotheses is efficiently decidable.
\end{abstract}


\section{Introduction}

\emph{Guarded Kleene Algebra with tests (GKAT)} is a formalism for reasoning about control flow in simple imperative \verb+while+ programs.
Its axiomatisation allows for equational reasoning, and it features an efficient automata-theoretic decision procedure for equivalence of programs~\cite{SmolkaFHKKS20}. 
GKAT has been extended to include restricted forms of non-local control flow~\cite{nonlocal}, probabilities~\cite{RozowskiKKS023}, weighted branching~\cite{weightedgkat}, and approximate reasoning~\cite{GomezBG25}.

GKAT programs are uninterpreted deterministic programs built from a base set of primitive programs and propositional tests.
For instance, the if-then-else statement ``if $b$ holds then execute the uninterpreted program $p$, otherwise execute $q$'', is expressed concisely as $p \mathbin{\purple{+}}_{\TEST b} q$.
The operator $\mathbin{\purple{+}}_{\TEST b}$ is a guarded choice operator, closely related to branching operations common in process calculi~\cite{SchmidR0R22}, e.g., probabilistic choice~\cite{StarkS00}.
GKAT is a fragment of \emph{Kleene algebra with tests (KAT)}~\cite{Kozen97}, likewise allowing algebraic reasoning about program equivalence.
One source of motivation to study a fragment of KAT rather than KAT itself is the complexity of deciding program equivalence: decidability in KAT is PSPACE-complete, but program equivalence in GKAT can be decided in nearly linear time~\cite{SmolkaFHKKS20}.

Like in KAT, GKAT programs are strictly propositional, so equivalence in GKAT is agnostic about the exact nature of each primitive program and test.
This results in a notion of equivalence that only sees the control-flow structure of programs.
In practice, however, one usually has access to further information about the (sub)programs at hand, and would like to incorporate such information when reasoning about program equivalence.

A motivating example for this paper arises from the wish to model notions of \emph{hoisting} in programming languages,
as expressed for instance by the following equation in a language with variable assignments~\cite{LiellCockS25}:
\begin{equation}
    \label{eq:first_example}
    \IF{\TEST b}{(x \leftarrow t);~ u}{(x \leftarrow t);~ v}
    \ \ =\ \  (x \leftarrow t) ; ~\IF{\TEST b}{u}{v}
\end{equation}
Above, \(u\) and \(v\) are arbitrary programs, \((x \leftarrow t)\) represents the assignment of a value \(t\) to the variable \(x\), and \(\TEST b\) is a Boolean expression.
An equation such as~\eqref{eq:first_example} typically holds when $x$ does not occur free in the guard $b$. 
In GKAT, the above hoisting equation is modelled by abstracting away the details of the actions that occur.
Specifically, the program action \((x \leftarrow t)\) is replaced by a formal \emph{action symbol} \(p\).
Doing the same for \(u\) and \(v\), if we let \(q\) represent \(u\) and \(r\) represent \(v\), we can express the equivalence in~\eqref{eq:first_example} as follows in GKAT:
\begin{equation}
    \label{eq:hoisting-gkat}
    pq \mathbin{\purple{+}}_{\TEST b} pr = p(q \mathbin{\purple{+}}_{\TEST b} r)
\end{equation}
For ordinary programs, in which the tests and constituent programs are interpreted,~\eqref{eq:hoisting-gkat} is only valid when the execution of the program \(p\) has no effect on the truth value of \(\TEST b\).
For example, in~\eqref{eq:first_example} this lack of effect occurs when the variable \(x\) does not occur free in the guard \(\TEST b\). 
We abstractly model this at the level of GKAT as follows: we assert the following two equations between programs.
\begin{equation}\label{eq:bp-invariant}
	{\TEST b}p = {\TEST b}p{\TEST b} 
    \quad\text{ and } \quad 
    \TEST{\overline{b}}p = \TEST{\overline{b}}p \TEST{\overline{b}}
\end{equation}
Here, a guard \(\TEST{b}\) appearing outside of a \(\purple{+}_{\TEST b}\) represents an assertion.
The first equation says that if $\TEST b$ holds before execution of $p$ then it also holds afterwards, and the second says the same for the complement $\TEST{\overline{b}}$.
Together, these equations assert that $p$ does not affect the validity of $\TEST{b}$.
Assuming~\eqref{eq:bp-invariant}, the hoisting equation \eqref{eq:hoisting-gkat} is valid in GKAT.

In this paper, we aim to support reasoning in GKAT under hypotheses such as~\eqref{eq:bp-invariant}.
In the general form $\TEST bp = \TEST bp \TEST c$, where \(\TEST b\) and \(\TEST c\) are tests and \(p\) is a primitive program (what we will call an \emph{action}), these equations are referred to as \emph{Hoare hypotheses}.
Hoare hypotheses are the propositional counterparts to \emph{Hoare triples}: these are formulas written \(\{P\}C\{Q\}\), where \(P\) and \(Q\) are interpreted tests and \(C\) is an interpreted program, and \(P\) is a sufficient starting condition to ensure \(Q\) holds after \(C\) is run~\cite{Hoare}.
Hoare hypotheses have already been considered in the context of Kleene algebra with tests~\cite{Kozen97,DBLP:conf/csl/KozenS96,Kozen00}, where decidability and completeness results are proven for the resulting equational theory.
In fact, in plain Kleene algebra, the theory of reasoning about equivalence under such hypotheses is well developed, with both completeness and decidability results for classes of hypotheses that go well beyond Hoare hypotheses~\cite{Cohen94,KozenM14,DoumaneKPP19,PousRW24}.

However, in the context of GKAT, the addition of hypotheses remains entirely open.
It is challenging even to define a suitable (sound) interpretation of programs in the presence of hypotheses, let alone decidability and completeness of the new theory.
Without hypotheses, GKAT programs are interpreted in the same way as KAT programs, as languages of words whose letters alternate between test and action symbols.
The general technique used in KA and KAT for adapting their semantics to account for hypotheses is to ``blow up'' the language by adding all words necessary to ensure soundness of the considered hypotheses~\cite{DoumaneKPP19,PousRW24}. 
Subsequently, an expression is identified that describes the resulting language, and the problem of checking equivalence then reduces to the equivalence of the resulting expressions representing this semantics under hypotheses.
This technique is not applicable to GKAT: GKAT programs are deterministic, so to be able to recover an expression that represents the semantics under hypotheses, the language has to remain deterministic.
The novelty of our approach to interpreting GKAT expressions in the presence of hypotheses is that it \emph{cuts away} all irrelevant words. 
Determinism is what allows our decision procedure to retain the nearly-linear time complexity of ordinary GKAT~\cite{SmolkaFHKKS20}, whereas a decision procedure based on KAT (with hypotheses) would likely be in PSPACE.

We take the first step towards developing a theory of \emph{GKAT with hypotheses}, focusing on what we call \emph{word hypotheses}.
Word hypotheses are a generalization of the propositional Hoare hypotheses studied in the context of KAT in~\cite{Kozen00}, and they state that for a particular alternation of programs and tests, a specific postcondition needs to hold.
Propositional Hoare hypotheses are of the form $\TEST bp = \TEST bp \TEST c$, also discussed above in the example related to hoisting.
Word hypotheses generalise this to equations of the form $\TEST {b_0} p_1 \TEST {b_1} p_2 \ldots \TEST {b_{n-1}} p_n = \TEST {b_0} p_1 \TEST{b_1} p_2 \ldots \TEST{b_{n-1}} p_n \TEST {b_n}$,
where all $\TEST{b_i}$'s are tests and $p_i$'s are programs.

We start with the simpler form of Hoare hypotheses in \cref{sec:hoare hypotheses}.
There, we introduce a language semantics, and prove completeness of our axiomatisation with respect to that semantics. 
We then move to the main contribution of this paper in \cref{sec:words and queues}: a decision procedure for general word hypotheses in nearly linear time. 
The key idea is to construct so-called \emph{queue automata}, which remember the last read tests and programs, and which can detect the violation of a hypothesis. 
This yields precisely the behaviour of GKAT programs under word hypotheses, where program executions that violate a hypothesis are cut away. 
By showing that these queue automata reduce to plain GKAT automata, we obtain a reduction of decidability of equivalence to the original automata-theoretic decision procedure for equivalence of GKAT programs from~\cite{SmolkaFHKKS20}, preserving its near linear time complexity. 
Basic notions and notation are covered in a preliminary section, \cref{sec:prelims}.

\subparagraph*{Related work.} 
For plain Kleene algebra, there is a well-developed tradition and theory of reasoning under additional hypotheses~\cite{KozenM14,DoumaneKPP19,Cohen94,PousRW24}, often focusing on completeness and decidability w.r.t.\ a model that is canonically defined from the given hypotheses.
Hypothesis elimination in the context of Kleene algebra is also studied in relational models; see, e.g.,~\cite{DBLP:conf/RelMiCS/Nakamura24}.
Further, as mentioned above, Hoare hypotheses in the context of Kleene algebra with tests and associated reasoning was systematically studied in~\cite{Kozen00}.
The current paper instead focuses specifically on the addition of hypotheses to GKAT and its model of guarded strings. 

Propositional Hoare triples are considered in~\cite{SmolkaFHKKS20}, in the context of a completeness theorem for GKAT. 
However, only the Hoare triples that are already valid in GKAT are considered, and no reasoning under new hypotheses is studied. 
The Hoare triples studied in this paper are typically \emph{not} valid in GKAT, but are additional assumptions made about specific programs.

Program logics like NetKAT~\cite{AndersonFGJKSW14} and Guarded NetKAT~\cite{GNetKAT} involve reasoning about programs in a similar fashion to KAT and GKAT, but with additional hypotheses formally imposed by insight into the nature of primitive actions and tests.
The desire to capture the structure of these program logics abstractly provides further motivation for the current work.


\section{Guarded Kleene Algebra with Tests} 
\label{sec:prelims}

\emph{Guarded Kleene Algebra with Tests} (\emph{GKAT}) is an abstraction from standard imperative programming: its basic elements are a collection of \emph{primitive programs} \(\Sigma\) and a finite collection of \emph{primitive tests} \(T\), and all other programs expressible in the language of GKAT are formed by combining these elements with \texttt{if-then-else} and \texttt{while-do} constructs.
Formally, the collection \(\GExp\) of \emph{GKAT expressions} is given by the BNF grammar
\begin{align*}
    \GExp \ni e, f 
        &::= b \in \BExp \mid p \in \Sigma \mid e +_b f \mid e\cdot f \mid e^{(b)} \\
    \BExp \ni b, c 
        &::= t \in T \mid \bar b \mid b \vee c
\end{align*}
Intuitively understood, each test \(b \in \BExp\) (thought of as a GKAT expression) is an assertion \(\ASSERT b\), the expression \(e +_{b} f\) is conditional branching \(\IF{b}{e}{f}\), \(e \cdot f\) is sequential composition, and \(e^{(b)}\) is the conditional loop \(\WHILE{b}{e}\).
For brevity, we use juxtaposition \(ef\) instead of \(\cdot\) when the meaning is clear.
We define \(b \wedge c = \overline{(\bar b \vee \bar c)}\), \(1 = b \vee \bar b\), and \(0 = b \wedge \bar b\). 
We also sometimes write $+_{bc}$ instead of $+_{b \wedge c}$.

Since we have assumed that \(T\) is finite, there are only finitely many tests modulo the Boolean algebra axioms, and this in particular implies that there is a well defined notion of \emph{atomic test}.
Specifically, an \emph{atomic test} is a Boolean expression of the form \(\sigma(t_1) \wedge \cdots \wedge \sigma(t_n)\), where \(T = \{t_1, \dots, t_n\}\) and \(\sigma(t_i) \in \{t_i, \overline{t_i}\}\) for each \(i \le n\). 
The set of atomic tests in \(\BExp\) is written \(\At\), and we use the Greek letters \(\alpha, \beta\) for atomic tests.
Thinking of tests as descriptions of program states (like how the test \(\mathtt{x} > 0\) describes program states where the variable \(\mathtt{x}\) takes a value above \(0\)), an atomic test represents a complete description of one specific program state, insofar as the status of the tests of the program affects its flow.
It is worth noting that the number of atomic tests is exponential in the number of primitive tests.

\subparagraph*{Language semantics.}
GKAT expressions are abstract representations of programs that have a well defined notion of program trace. 
Just like in KAT~\cite{Kozen97}, these traces take the form of \emph{guarded strings}, which are alternating sequences of atomic tests and primitive programs.
For a set $X$ we denote by $X^*$ the set of all words over $X$, with the empty word denoted by $\varepsilon$.

\begin{definition}
    A \emph{guarded string} is a word \(w = \alpha_0 p_1 \alpha_1 \cdots \alpha_{n-1} p_n \alpha_n\) whose letters alternate between atomic tests and primitive programs, i.e.,  \(\alpha_0, \dots, \alpha_n \in \At\) and \(p_1, \dots, p_n \in \Sigma\). 
    Formally, the set of guarded strings is \((\At \cdot \Sigma)^* \cdot \At\).
    A \emph{language} of guarded strings, or \emph{guarded language}, is a set of guarded strings. 
    We write \(\GL\) for the set of guarded languages \(\GL = \mathcal P((\At \cdot \Sigma)^*\cdot \At)\). 
\end{definition}

Equating program traces with guarded strings leads to the standard interpretation of GKAT expressions, the \emph{language semantics}.
To state the language semantics precisely, we need to define the product of two guarded languages $L_1, L_2 \in \GL$. 
To this end, we define the \emph{product} $\diamond$ of two guarded strings \(u\alpha, \beta v \in (\At \cdot \Sigma)^*\cdot \At\) to be $u\alpha v$ if \(\alpha = \beta\) and undefined otherwise. 
This is then lifted to languages:
\begin{equation}
    L_1 \diamond L_2 = \left\{
        u\alpha \diamond \alpha v
        ~\middle|~ 
           u\alpha \in L_1, \text{ }
            \alpha v \in L_2
    \right\}
\end{equation}
Moreover, we define an associated Kleene star of a guarded language $L$ by $L^\star = \bigcup_{n \geq 0} L^{\langle n \rangle }$ where $L^{\langle 0 \rangle} = \At$
and $L^{\langle n+1 \rangle} = L \diamond L^{\langle n \rangle}$.

We write \(\BA\) for the equational theory of Boolean algebras, \(=_\BA\) for the congruence generated by \(\BA\), and \(\le_\BA\) for the order induced by \(\BA\), i.e., \(b \le_\BA c\) \emph{iff} \(b \vee c =_\BA c\).
For ease of notation, we also define \(b \diamond L = \{\alpha \in \At \mid \alpha \le_{\BA} b\} \diamond L\) and \(L \diamond b = L \diamond \{\alpha \in \At \mid \alpha \le_{\BA} b\}\). 
Later, we will make use of the fact that $(b \diamond L) \cap K = b \diamond (L \cap K)$ and $(L \diamond b) \cap K = (L \cap K) \diamond b$ for any guarded languages $L,K$ and
$b \in \BExp$.

\begin{definition}
    The \emph{language interpretation} (or \emph{language semantics}) \(\mathcal L(e)\) of a GKAT expression \(e \in \GExp\) is defined recursively on \(e \in \GExp\) by 
    \begin{gather*}
        \L(b) = \{\alpha \in \At \mid \alpha \le_\BA b\}
        \quad 
        \L(p) = \{\alpha p \beta \mid \alpha,\beta \in \At\}
        \\
        \L(e +_b f) = b\diamond \L(e) \cup \bar b\diamond \L(f)
        \quad 
        \L(ef) = \L(e)\diamond \L(f)
        \quad 
        \L(e^{(b)}) = (b \diamond \L(e))^{\star} \diamond \bar b
    \end{gather*}
    Thus, we obtain the \emph{semantic map} \(\mathcal L \colon \GExp \to \GL\).
    Given \(e, f\in \GExp\), if \(\L(e) = \L(f)\), we say that \(e\) and \(f\) are \emph{language equivalent}. 
\end{definition}

The following theorem is~\cite[Theorem 5.10]{SmolkaFHKKS20}, which states the decidability and time complexity of deciding language equivalence of GKAT expressions, assuming that the sets of atomic tests \(\At\) and alphabet \(\Sigma\) are held constant.
Given an expression \(e \in \GExp\), we write \(|e|\) for the number of symbols that appear in \(e\).

\begin{theorem}[Decidability in Nearly-Linear Time~\cite{SmolkaFHKKS20}] 
    \label{thm:gkat decidability}
    Let \(e,f \in \GExp\). 
    For \(|\At|\) held constant, it is decidable in \(\bigO(n\alpha(n))\)-time whether \(\L(e) = \L (f)\), where \(\alpha(n)\) is the inverse Ackermann function and $n = |e|+|f|$. 
\end{theorem}

The complexity class \(\bigO(n\alpha(n))\) is called \emph{nearly-linear} because of the extremely slow growth of the inverse Ackermann function \(\alpha(n)\). 
The algorithm in~\cite{SmolkaFHKKS20} and its complexity analysis relies on~\cite{hopcroft1971linear,Tarjan75}, and requires explicit calculation of \(\At\) and \(\Sigma\) ahead of time.
Without the assumptions that \(|\At|\) and \(\Sigma\) are held fixed, the automata generated by GKAT expressions incur an exponential blow-up (because \(|\At| = 2^{|T|}\)).
This is a known limitation of existing decision procedures, and is implicit even in earlier work on KAT~\cite{CohenKozenSmithKAT}.
Recently developed symbolic methods (going back to~\cite{Pous15}) appear to avoid this issue in practice; see~\cite{ZhangFJVSG25}.

\subparagraph*{Axiomatisation.}
One of the key features of GKAT is that it allows for equational reasoning between programs. 
In~\cite{SmolkaFHKKS20}, a sound and complete axiomatisation of language equivalence for GKAT expressions is presented.
We refer to that axiomatisation as $\GKAT$, and record the axioms (other than the Boolean algebra axioms \(\BA\) and equational logic) in \cref{fig:axioms}.
We write $\GKAT \vdash e = f$ if the equality $e=f$ is derivable from the axioms in $\GKAT$.

Axioms (W3) and (UA) in \cref{fig:axioms} use the map \(E \colon \GExp \to \BExp\) defined recursively by
\begin{gather*}
    E(b) = b 
    \qquad 
    E(p) = 0
    \qquad 
    E(e +_b f) = bE(e) \vee \bar bE(f)
    \\
    E(ef) = E(e)\wedge E(f)
    \qquad 
    E(e^{(b)}) = \bar b
\end{gather*}
for any \(e \in \GExp\). 
Intuitively, \(E(e)\) is the least condition that guarantees \(e\) immediately halts and accepts, so \(E(e) =_\BA 0\) means that the program $e$ performs an action under any circumstance.
In terms of languages, \(E(e) = \bigvee (\At \cap \L(e))\).
The side condition \(E(e) =_\BA 0\) is analogous to the side condition regarding the \emph{empty word property} in Salomaa's axioms for the algebra of regular languages~\cite{Salomaa66}.
These side conditions are necessary to ensure soundness of (W3) and (AU):
for example, by (W1), \(\GKAT \vdash 1 = 1 \cdot 1 +_1 1\), but \(\L(1) \neq \L(1^{(1)}) = \emptyset\)~\cite{SmolkaFHKKS20}.

\begin{theorem}[\cite{SmolkaFHKKS20}]
	For any $e,f \in \GExp$, \(\GKAT \vdash e = f\) iff $\L(e) = \L(f)$.
\end{theorem}

Note the presence of the axiom (UA). 
This axiom is an extenson of (W3) to arbitrary systems of equations, ensuring that solutions to systems (satisfying the side condition) are unique when they exist.
It is still open whether (UA) is derivable from the other axioms~\cite{SchmidKS23}.

\begin{figure}[t]
    \small
    \begin{gather*}
        \begin{gathered}
            \text{\bf Union Axioms}\hfill \\
            \begin{aligned}
                (\text{U1}) && e +_b e &= e \\
                (\text{U2}) && e +_b f &= f +_{\bar b} e \\
                (\text{U3}) && (e +_b f) +_c g &= e +_{bc} (f +_c g) \\
                (\text{U4}) && e +_b f &= be +_b f \\
                (\text{U5}) && eg +_b fg &= (e +_b f)g \\
                \phantom{()} 
            \end{aligned}
        \end{gathered}
        \qquad 
        \begin{gathered}
            \text{\bf Sequencing Axioms}\hfill \\
            \begin{aligned}
                (\text{S1}) && (ef)g &= e(fg) \\
                (\text{S2}) && 0e &= 0 \\
                (\text{S3}) && e0 &= 0 \\
                (\text{S4}) && 1e &= e \\
                (\text{S5}) && e1 &= e \\
                (\text{S6}) && bc &= b \wedge c
            \end{aligned}
        \end{gathered}
        \\
        \begin{gathered}
            \text{\bf Loop Axioms}\hfill \\
            \begin{aligned}
                (\text{W1}) && e^{(b)} &= ee^{(b)} +_b 1 \\
                (\text{W2}) && (e +_c 1)^{(b)} &= (ce)^{(b)} \\
            \end{aligned}
            \hspace{1em} (\text{W3}) \quad \inferrule{g = eg +_b f \quad E(e) =_{\BA} 0}{g = e^{(b)}f}
        \end{gathered}\
        \\
        (\text{UA}) \quad \inferrule{
            \parbox{16em}{\(\begin{gathered}
                \{g_i = e_{1i}g_1 +_{b_{1i}} \cdots e_{ni}g_n +_{b_{ni}} d_i\}_{i \le n}  
                \\ 
                \{f_i = e_{1i}f_1 +_{b_{1i}} \cdots e_{ni}f_n +_{b_{ni}} d_i\}_{i \le n}
            \end{gathered}\)}
                \quad \{b_{ij} b_{ik} =_\BA 0\}_{j \neq k, i \le n} 
                \quad \{E(e_{ij}) =_{\BA} 0\}_{i,j\le n}
            }{g_i = f_i}
    \end{gather*}
    \caption{\label{fig:axioms}\cite{SmolkaFHKKS20}
        Axioms for proving program equivalence of GKAT-expressions. 
        In the GKAT expressions that appear above, \(e,f,g,e_{ij},f_i,g_i \in \GExp\) and \(b,c,b_{ij},d_i \in \BExp\).
    }
\end{figure}

\section{Hoare hypotheses}
\label{sec:hoare hypotheses}

In this section, we incorporate \emph{Hoare hypotheses} into GKAT (see \cref{def:hoare hypothesis}).
We introduce their semantics, prove completeness, and obtain decidability as a consequence of the underlying construction of the completeness proof.
An analysis of the complexity of this decision procedure appears at the end of the section, but it is worth noting in advance that it is not as efficient as the decision procedure we introduce later in \cref{sec:words and queues} for word hypotheses.

Let us begin with the general definition of \emph{hypothesis} in the context of GKAT.

\begin{definition}
    \label{def:hypothesis}
    A \emph{hypothesis} is a formal equation \(e = f\), where \(e,f \in \GExp\). 
    For a set $H$ of hypotheses, we write $\GKAT + H \vdash e=f$ when $e=f$ is derivable from the axioms in \(\GKAT\) and the hypotheses in $H$.
\end{definition}

Intuitively, hypotheses extend our knowledge of the meanings of GKAT programs by imposing additional equations.
Following this intuition, \emph{Hoare hypotheses}, formally introduced below, can be thought of as statements of the form, ``if the test \(b\) holds before executing the primitive program \(p\), then the test \(c\) holds after \(p\) executes''.

\begin{definition}[Hoare hypothesis]
    \label{def:hoare hypothesis}
    A \emph{Hoare hypothesis} is a hypothesis of the form \(bp = bpc\), where \(b, c \in \BExp\) and \(p \in \Sigma\). 
    If \(H\) is a set of Hoare hypotheses, then we say that \(\GKAT+H\) is a \emph{Hoare extension} of \(\GKAT\).
\end{definition}

Simple as they are, Hoare hypotheses in GKAT are able to capture a number of different ways in which imperative programs might interact.
To illustrate the expressiveness of Hoare hypotheses in GKAT, let us return to our motivating application: a phenomenon in imperative programming called \emph{hoisting}, which allows one to ``hoist'' a program that appears at the start of each branch of an \(\IF{\!{-}\!}{\!{-}\!}{}\) construct out to just before the conditional, as in~\eqref{eq:first_example}.
Given a program \(p\) and a test \(b\), the \emph{hoisting principle for \(p\) over \(b\)} is the set of hypotheses 
\[
    \mathsf{Hst}_{p,b} = \{p(e +_b f) = pe +_b pf \mid e,f \in \GExp\}
\]
As discussed in the introduction, these hoisting principles are not already valid in \(\GKAT\).
However, they become valid with the addition of two Hoare hypotheses that state the preservation of the test \(b\) by the execution of the primitive program \(p\).
Formally, the \emph{invariant test hypothesis for \(b\) under \(p\)} is the set of Hoare hypotheses
\[
    \mathsf{Inv}_{p,b} = \{bp = bpb, \bar b p = \bar b p \bar b\}
\]
In fact, invariance is equivalent to hoisting.

\begin{restatable}{lemma}{lemmaInvarianceEquivToHoisting}
    Let $e,f \in \GExp$. $\GKAT + \mathsf{Hst}_{p,b} \vdash e=f$ 
        if and only if $\GKAT+\mathsf{Inv}_{p,b}\vdash e=f$.
\end{restatable}

Each hypothesis \(bp = bpb\) is essentially the statement that the validity of \(b\) is ``maintained'' through the execution of \(p\).
Intuitively, this could also be expressed as the statement that running \(p\) before the assertion \(\ASSERT b\) has the same effect as running it after; i.e., that \(b\) and \(p\) \emph{commute}. 
The \emph{commutative principle for \(b\) and \(p\)} is the set of two hypotheses
\[
    \mathsf{Com}_{p,b} = \{pb=bp, p \bar b = \bar b p\}
\]
It turns out that these are also equivalent to invariant test hypotheses.

\begin{restatable}{lemma}{lemmaCommuteEquivToInvariance}
Let $e,f\in \GExp$. $\GKAT+\mathsf{Inv}_{p,b} \vdash e=f $
        if and only if $\GKAT+\mathsf{Com}_{p,b} \vdash e=f$.
\end{restatable}

\subparagraph*{Semantics and Completeness.}
Given a Hoare extension \(\GKAT + H\) of GKAT, we provide a semantics $\L_H(-)$ such that \(\GKAT + H \vdash e=f\) iff $\L_H(e)=\L_H(f)$.
The idea behind our semantics for GKAT with Hoare hypotheses is to adjust the ordinary language semantics of GKAT by restricting it to the guarded strings that satisfy the hypotheses. 
Later, we will show how to reduce semantic equivalence to ordinary language equivalence of GKAT and appeal to \cref{thm:gkat decidability} for decidability.

Intuitively, the language semantics of a GKAT expression \(e\) (in general) consists of all program traces that are possible during an execution of \(e\). 
The basic idea, then, is that under a hypothesis \(bp = bpc\), there are fewer of these possible program traces.
In particular, if a trace contains a step \(\alpha_ip\) with \(\alpha_i \le_{\BA} b\), then the execution of \(p\) must proceed with an \(\alpha_{i+1}\) in which \(c\) holds.
Our semantics therefore removes from the ordinary language semantics all traces that do not satisfy this condition. 

\begin{definition}
    \label{def:language up to H}
    Let \(H\) be a set of Hoare hypotheses.
    Given \(e \in \GExp\), we define the \emph{language semantics up to \(H\)} of \(e\) to be the guarded language
    \[
        \L_H(e) = \L(e)\cap \left\{ \alpha_0p_1\cdots\alpha_{n-1}p_n\alpha_{n} ~\Big|~ 
        \begin{gathered}
            \text{for all }i \le n, \text{if }
            b p_i = b p_i c\in H\\
            \text{ and } \alpha_i\leq_{\BA} b \text{, then } \alpha_{i+1}\leq_{\BA} c
        \end{gathered}~    
        \right\} 
    \]
    Given \(e,f \in \GExp\), we say that \(e\) and \(f\) are \emph{language equivalent up to \(H\)} if \(\mathcal L_H(e) = \mathcal L_H(f)\).
\end{definition}

It will be useful to observe that $\L_H$ is a homomorphism, in the following sense.
\begin{restatable}{lemma}{lemmaLHHom}\label{lem:lh-hom}
	For any set of Hoare hypotheses $H$, $e,f \in \GExp$ and $b \in \BExp$:
	\begin{align*}
		\L_H(b) &= \{\alpha \in \At \mid \alpha \le_\BA b\} & \L_H(ef) &= \L_H(e)\diamond \L_H(f) \\ 
		\L_H(e +_b f) &= b\diamond \L_H(e) \cup \bar b\diamond \L_H(f) & \L_H(e^{(b)}) &= (b \diamond \L_H(e))^{\star} \diamond \bar b
	\end{align*}
\end{restatable}

The axioms of GKAT with Hoare hypotheses are sound with respect this semantics.

\begin{restatable}[Soundness]{theorem}{theoremSoundness}
    \label{thm:soundness lang}
    Let \(H\) be a set of Hoare hypotheses.
    For any \(e,f \in \GExp\), 
    if \(\GKAT + H \vdash e = f\), then \(\L_H(e) = \L_H(f)\).
\end{restatable}

To prove completeness, we adapt a standard approach from Kleene algebra with hypotheses (e.g.,~\cite{PousRW24}) to the GKAT setting: we construct a \emph{reduction}.
That is, we define a syntactic translation of GKAT expressions that incorporates the hypotheses into any given expression.
We start by defining the reduction for a single hypothesis.

\begin{definition}
	Let \(h\) be the Hoare hypothesis \(bp = bpc\). 
	The \emph{translation} function $\transl{h} \colon \GExp \rightarrow \GExp$ is the homomorphic extension of the map defined by
	\begin{gather*}
		\transl{h}(d) = d 
		\qquad 
		\transl{h}(p) = pc +_b p
		\qquad 
		\transl{h}(q) = q 
	\end{gather*}
	where \(d \in \BExp\) and \(q \in \Sigma\) with \(p \neq q\).
\end{definition}

Given a hypothesis $h$, the translation $\transl{h}(e)$ of an expression $e$ is provably equivalent to $e$ up to the axioms of $\GKAT$ and the hypothesis $h$.

\begin{lemma}
    \label{lem:transl equiv}
    For any \(e \in \GExp\), \(\GKAT+h \vdash \transl{h}(e) = e\).
\end{lemma}

\begin{proof}
    By induction on \(e\). 
    The inductive cases are covered because $\transl{h}$ is a homomorphism. 
    The only interesting base case is where \(e = p\) and where the hypothesis \(h\) is \(bp = bpc\). 
    We have the following derivation, in which we use (U4), (\(h\)), (U4), and (U1), in that order:
    \begin{align*}
        \GKAT+h \vdash 
        \transl{h}(p)
        = pc +_b p
        = bpc +_b p 
        = bp +_b p 
        = p +_b p 
        = p \tag*{\qedhere}
    \end{align*}
\end{proof}

We are now in a position to define the reduction for multiple Hoare hypotheses.

\begin{definition}
    Let \(H\) be a set of Hoare hypotheses.
    The translation $\transl{H} \colon \GExp \rightarrow \GExp$ of \(e \in \GExp\) is defined inductively by \[
        \transl{H}(e) = \transl{h}(\transl{H\setminus h}(e))
    \]
    where \(H\setminus h\) is the set of Hoare hypotheses \(H\) excluding \(h\).
\end{definition}

This is well defined because of the following lemma, which shows that the order in which the translations for individual Hoare hypotheses in $H$ are applied is immaterial.

\begin{lemma}
    \label{lem:hypcommute}
    For any Hoare hypotheses \(h_1\) and \(h_2\), the translations commute:
    \[
        \GKAT + H \vdash \transl{h_1}\transl{h_2}(e)
        = \transl{h_2}\transl{h_1}(e).
    \]
\end{lemma}

\begin{proof}
    Suppose \(h_i\) is the hypothesis \(b_ip_i = b_ip_ic_i\) for \(i=1,2\).
    We proceed by induction on \(e\).
    The only interesting case is when \(e = p_1 = p_2\), since the other cases are either trivial applications of the induction hypothesis or the two translations do not interact.
    So, let \(p = p_1 = p_2\).
    In the calculation below, we will use the identity \(\GKAT \vdash e +_b (f +_c g) = (e +_b f) +_{b \vee c} g\), which we denote (U3').
    This is also (U3') in~\cite{SmolkaFHKKS20}, where it is proven by repeatedly using (U2) and (U3) and the Boolean algebra axioms.
    We have 
    \begin{align*}
        \GKAT + H \vdash \transl{h_1}\transl{h_2}(p)
        &= \transl{h_1}(pc_2 +_{b_2} p)                                      
            \tag{def.~\(T_{h_2}\)}
            \\
        &= (pc_1 +_{b_1} p)c_2 +_{b_2} (pc_1 +_{b_1} p)                      
            \tag{def.~\(T_{h_1}\)}
            \\
        &= (pc_1c_2 +_{b_1} pc_2) +_{b_2} (pc_1 +_{b_1} p)                   
            \tag{U5}
            \\
        &= pc_1c_2 +_{b_1b_2} (pc_2 +_{b_2} (pc_1 +_{b_1} p))                
            \tag{U3}
            \\
        &= pc_1c_2 +_{b_1b_2} ((pc_2 +_{b_2} pc_1) +_{b_1\vee b_2} p)               
            \tag{U3'}
            \\
        &= pc_1c_2 +_{b_1b_2} ((pc_1 +_{\bar b_2} pc_2) +_{b_1\vee b_2} p)               
            \tag{U2}
            \\
        &= pc_1c_2 +_{b_1b_2} (pc_1 +_{b_1} (pc_2 +_{b_1\vee b_2} p))               
            \tag{U3, \(\BA\)}
            \\
        &= pc_1c_2 +_{b_1b_2} ((pc_1 +_{b_1} pc_2) +_{b_1\vee b_2} p)              
            \tag{U3', \(\BA\)}
            \\
        &\hspace{4em}{\vdots} \tag{same steps backwards}\\
        &= \transl{h_2}\transl{h_1}(p) \tag*{\qedhere}
    \end{align*}
\end{proof}

A routine induction on the number of hypotheses allows us to extend \cref{lem:transl equiv} to an arbitrary finite set of Hoare hypotheses \(H\).

\begin{theorem}
    \label{theorem:syntactic-equivalence}
    For any \(e \in \GExp\), \(\GKAT+H \vdash \transl{H}(e) = e\). 
\end{theorem}

The following relationship between the language model of \(\GKAT\) and the language model \(\L_H\) of \(\GKAT + H\) is core to the proof that \(\transl{H}\) is a meaningful translation of GKAT expressions.

\begin{theorem}\label{lemma:semantic-equivalence}
    Let $e\in \GExp$. 
    Then $\L_H{(e)}=\L{(\transl{H}(e))}$.
\end{theorem}

\begin{proof}
    We proceed by induction on $e$. 
    If $e=b\in\BExp$, the result trivially holds because $\transl{H}(b)=b$ and $\L_H{(b)}=\L(b)$ by definition. 

    For $e = p\in\Sigma$, we proceed by induction on the number of hypotheses in $H$ that involve $p$. 
	If there are no hypotheses in $H$ that involve $p$, the result trivially holds. 
    Otherwise, pick any $h\in H$ of the form $bp=bpc$ for some $b,c\in\BExp$.
    As the order of translations in $\transl{H}$ does not matter (\cref{lem:hypcommute}), we assume the translation $\transl{h}$ is performed `last', i.e. $\transl{H}(e)=\transl{ h}(\transl{ H\setminus h}(e))$. 
	Our induction hypothesis states that $\L_{H\setminus h}(p)=\L(\transl{H\setminus h}(p))$.
	We prove  \(\L_H(p) = \L(T_H(p))\), as follows:
    \begin{align*}
		\L_H(p)&= \L_H(bp +_b \bar b p) \tag{basic property of $\L$, therefore also $\L_H$} \\
		&= (b\diamond \L_{H}(p))\cup (\bar b \diamond \L_{H}(p)) \tag{\cref{lem:lh-hom}} \\
		&= (b\diamond \L_{H\setminus h}(p) \diamond c)\cup (\bar b \diamond \L_{H\setminus h}(p)) \tag{from def.~$\L_{H}$, $\L_{H\setminus h}$} \\
		&= (b\diamond \L(\transl{H\setminus h}(p))\diamond c)\cup (\bar b \diamond \L(\transl{H\setminus h}(p))) \tag{Induction Hypothesis}\\
		&=(b\diamond \L(\transl{H\setminus h}(p)c))\cup (\bar b \diamond \L(\transl{H\setminus h}(p))) \tag{def.~\(\L\)}\\
		&= \L(\transl{H\setminus h}(p)c+_b \transl{H\setminus h}(p)) \tag{def.~\(\L\)}\\
		&=\L(\transl{H\setminus h}(pc+_b p)) \tag{def.~$\transl{H\setminus h}$}\\
		&= \L(\transl{H\setminus h}(\transl{h}(p))) \tag{def.~\(\transl{h}\)}\\
		&= \L(\transl{h}(\transl{H\setminus h}(p))) \tag{\cref{lem:hypcommute}} \\
		&= \L(\transl{H}(p))
    \end{align*}
	This concludes the case $e= p \in \Sigma$.
	
	For the induction step $e +_b f$, we get:
	\begin{align*}
		\L_H(e+_b f)
		&= b \diamond \L_H(e)\cup \bar b\diamond \L_H(f) \tag{\cref{lem:lh-hom}}\\
		&= b\diamond \L(\transl{H}(e))\cup \bar b \diamond \L(\transl{H}(f)) \tag{Induction Hypotheses}\\
		&= \L(\transl{H}(e)+_b \transl{H}(f)) \tag{def.~\(\L\)}\\
		&= \L(\transl{H}(e+_b f)) \tag{def.~\(\transl{H}\)}
	\end{align*}
	The other inductive cases $e \diamond f$ and $e^{(b)}$ can be proven similarly from~\cref{lem:lh-hom}.
\end{proof}

\cref{theorem:syntactic-equivalence,lemma:semantic-equivalence} lead us to our completeness and decidability theorems.
Together, they allows us to reduce equivalence of $\GKAT + H$ to $\GKAT$ equivalence.

\begin{theorem}[Completeness]
    \label{thm:hoare completeness}
    Let $e,f\in\GExp$. 
    If $\L_H(e)=\L_H{(f)}$, then $\GKAT + H \vdash e=f$. 
\end{theorem}

\begin{proof}
    Let \(e,f \in \GExp\). 
    Using \cref{lemma:semantic-equivalence}, completeness of \(\GKAT\), and \cref{theorem:syntactic-equivalence},
    \begin{align*}
        \L_H(e)=\L_H{(f)} 
        &\Leftrightarrow \L{(\transl{H}(e))}=\L{(\transl{H}(f ))} \\
        &\Rightarrow \GKAT\vdash \transl{H}(e)=\transl{H}(f)  \\
        &\Leftrightarrow \GKAT+H\vdash e=f.  \tag*\qedhere
    \end{align*}
\end{proof}


\begin{theorem}[Decidability]
    \label{thm:hoare decidable}
    Let $e,f\in\GExp$. 
    It is decidable whether $\L_H(e)=\L_H(f)$. 
\end{theorem}

\begin{proof}
    To decide whether $\L_H(e)=\L_H(f)$, via \Cref{lemma:semantic-equivalence} it suffices to decide whether $\L(\transl{H}(e))=\L(\transl{H}(f))$. 
	This can be done using \Cref{thm:gkat decidability} and that $\transl{H}$ is computable.
\end{proof}

Let us take a moment to remark on the complexity of the decision procedure suggested by the completeness proof above.
The decision procedure reduces the problem to deciding ordinary language equivalence of two GKAT expressions \(e,f\).
We know from \cref{thm:gkat decidability} that equivalence is decidable in \(\mathcal O(n\alpha(n))\) time, where \(n = |e| + |f|\).
In the presence of a set of Hoare hypotheses \(H\), the decision procedure described in our completeness proof would therefore run in \(\mathcal O(m\alpha(m))\) time, where \(m = |\transl H(e)| + |\transl H(f)|\).
To measure the complexity of this decision procedure, we need only to compute the sizes of $\transl{H}(e)$ and \(\transl H (f)\).
We write $p\in e$ to express that $p\in\Sigma$ occurs in $e$, and let $n_p^H$ be the number of hypotheses in $H$ that contain a given action $p$. 

\begin{restatable}{lemma}{translationsize}\label{lem:translation size}
    Let $e\in\GExp$ and let $H$ be a set of Hoare hypotheses. 
	Then $$|\transl{H}(e)|=\sum_{p\in e}(3(2^{n_p^H})-3)+|e|$$
\end{restatable}

It follows that the decision procedure suggested by the completeness proof runs in exponential time in the worst case.
Fortunately, a much more efficient decision procedure exists that applies to a slightly broader class of hypotheses, although the same completeness proof technique does not appear to apply.
The larger class of hypotheses and their corresponding decision procedure is the subject of the next section.

\section{Decidability for Word Hypotheses}
\label{sec:words and queues}

In this section we consider a generalisation of Hoare hypotheses that we call \emph{word hypotheses}. 
We present a decision procedure for program equivalence under word hypotheses that maintains the nearly linear time complexity of GKAT. 
The decision procedure is based on a construction that adds a queue-like data structure (called an \emph{inspection queue}) to \emph{GKAT automata}, which are the operational models of GKAT (called \emph{\(G\)-coalgebras} in~\cite{SmolkaFHKKS20}).
 
As we will see, Hoare hypotheses are all instances of word hypotheses, so the presented decision procedure is a more efficient decision procedure for Hoare hypotheses than the one suggested by \cref{thm:hoare decidable}. 
Essentially, the automata-based decision procedure avoids computing the reduction map $\transl{H}$, which led to an exponential blow-up.

\subparagraph*{Word hypotheses.}
Recall that a guarded string is an element of \((\At \cdot \Sigma)^* \cdot \At\).
We call an element of \((\At \cdot \Sigma)^*\) a \emph{pre-guarded string}, because each of these is a prefix of a guarded string.
Roughly, a \emph{word hypothesis} is a constraint on the guarded strings admissible in a language that begin with a given pre-guarded string.

\begin{definition}
    \label{def:word hypotheses}
    A \emph{word hypothesis} is a formal identity of the form \(w = wc\) for some alternating concatenation of tests and atomic programs \(w \in (\BExp \cdot \Sigma)^*\) and some test \(c \in \BExp\).
    A guarded string \(\alpha_0 p_1 \alpha_1 \cdots \alpha_{l-1} p_l \alpha_l\) \emph{violates} \(w = wc\) if \(w = b_0 p_1 b_1 \cdots b_{l-1} p_l\), \(\alpha_i \le_{\BA} b_i\) for all \(i < l\), and \(\alpha_l \le_\BA \bar c\).
    A guarded string \emph{violates} a set of word hypotheses \(H\) if it violates a word hypothesis in \(H\).
\end{definition}

\begin{example}
In programming terms, a word hypothesis states a certain relationship between runs of a program and a post-condition.
For a simple concrete example, consider the following alternating sequence of tests and programs involving integer variables \(\mathtt{x}\) and \(\mathtt{y}\): 
\begin{equation}
    \label{eq:concrete eg}
    w = \mathtt{\blue{(y > 1)}(x \leftarrow x * y)\blue{(x > 0)}(x \leftarrow x * y)}
\end{equation}
The bracketed expressions involving the \(\mathtt{>}\) symbol are tests and the expressions involving \(\leftarrow\) are assignments, so \(w\) can appear on the left-hand-side of a word hypothesis.
Since \(\mathtt{x}\) and \(\mathtt{y}\) are integer variables, \(\mathtt{\blue{(x > 0)}}\) and \(\mathtt{\blue{(y > 1)}}\) imply that the values of both are at least \(1\) and \(2\) respectively. 
One word hypothesis that would be reasonable in this example is \(w = w\mathtt{\blue{(x > 4)}}\), since running \(w\) multiplies \(\mathtt{x}\) (which is at least \(1\)) by \(\mathtt{y}\) (which is at least \(2\)) twice.
Given any atomic test \(\blue{\alpha} \in \At\) that implies both \(\mathtt{\blue{(x == 1)}}\) and \(\mathtt{\blue{(y == 2)}}\), the guarded string 
\[
    \mathtt{\blue{\alpha} (x \leftarrow x * y) \blue{\alpha} (x \leftarrow x * y) \blue{\alpha}}
\]
violates \(w = w\mathtt{\blue{(x > 4)}}\), because \(\blue{\alpha}\) passes the tests \(\blue{(x > 0)}\) and \(\blue{(y > 1)}\), but it fails \(\blue{(x > 4)}\).
\end{example}

Given a guarded string \(w = \alpha_0 p_1 \cdots p_n \alpha_n\), we call a consecutive string of characters \(\alpha_i p_{i+1}  \cdots  p_j \alpha_{j}\), for \(0 \le i < j \le n\), a \emph{guarded substring} of \(w\).

\begin{definition}
    \label{def:language up to word hyp}
    Let $H$ be a set of word hypotheses.
    Given \(e \in \GExp\), we define
    \[
        \L_H(e) = \L(e)\cap \{ w\mid \text{no guarded substring of } w \text{ violates H}\} 
    \]
    The set of guarded strings \(\L_H(e)\) is the \emph{language semantics up to \(H\)} of \(e\).
    For any \(e,f \in \GExp\), if \(\L_H(e) = \L_H(f)\), then we say that \(e\) and \(f\) are \emph{language equivalent up to \(H\)}.
\end{definition}

Let us define the \emph{length} \(\len(w)\) of a (pre)guarded string \(w\) to be the number of primitive programs that appear in it.
That is, \(\len(\alpha_0 p_1 \cdots p_n \alpha_n) = n = \len(\alpha_0 p_1 \cdots p_n)\).
We take the length \(\len(w = wc)\) of a word hypothesis \(w = wc\) to be the length of \(w\).
In this terminology, Hoare hypotheses are word hypotheses of length \(1\), and the semantics of GKAT expressions up to word hypotheses is a direct generalisation of the semantics of GKAT expressions up to Hoare hypotheses from \cref{def:language up to H}.
This justifies our reuse of the notation \(\L_H\) for both notions of language semantics.
Our next (and final) task of the paper is to show that language equivalence up to \(H\) is efficiently decidable.

\subparagraph*{Inspection Queues.}
The core of our automata-based decision procedure is the notion of a \emph{bounded inspection queue}, which is a data structure built on the set of nonempty pre-guarded strings of some bounded length.
Intuitively, a bounded inspection queue is a queue with bounded capacity that evicts the head of the queue when an item is enqueued and the queue is at maximum capacity.
To state the definition precisely, we write \((\At \cdot \Sigma)^{\le l}\) for the set of pre-guarded strings of length at most \(l\) and define the following two functions.
\begin{itemize}
    \item 
    The \emph{inspection function} is the map \(\insp \colon (\NN \cup \{\infty\}) \times (\At \cdot \Sigma)^* \to (\At \cdot \Sigma)^*\) defined by
    \[
        \insp(n, w) 
        = \begin{cases}
            v & \text{\(w = uv\) and \(\len(v) = n\)} \\
            w & \text{\(\len(w) < n\)}
        \end{cases} 
    \]
    for any \(n \in \NN \cup \{\infty\}\) and pre-guarded string \(w\).
    Intuitively, \(\insp(n, -)\) returns the last \(n\) atomic test-action pairs in the input string.

    \item 
    For any given \emph{queue length} \(l \in \NN \cup \{\infty\}\), the \emph{\(l\)-bounded enqueue function} is the map \(\enq_l \colon (\At \cdot \Sigma) \times (\At \cdot \Sigma)^* \to (\At \cdot \Sigma)^{\le l}\) defined by 
    \(
        \enq_l(\alpha p, w) = \insp(l, w\alpha p)
    \).
    Intuitively, \(\enq_l\) is the enqueueing map for a \emph{bounded queue} of height \(l\).
\end{itemize}
For an easy example, \(\insp(\infty, w) = w\) always.
For a less trivial example, in the process of \(2\)-bounded enqueueing \(\beta q\) to the pre-guarded string \(\alpha_0 p_1 \alpha_1 p_2\) (which is of length \(2\)), the \(\alpha_0p_1\) portion of the queue is deleted, and the state of the queue becomes the string \(\alpha_1 p_2 \beta q\).
\[
    \enq_2(\beta q, \alpha_0 p_1 \alpha_1 p_2 ) 
    = \insp(2, {\alpha_0 p_1}\alpha_1 p_2 \beta q)
    = \alpha_1 p_2 \beta q .
\]
Our operational semantics for GKAT with word hypotheses essentially attaches a bounded queue to the standard operational model of GKAT expressions.

\begin{definition}[\cite{SmolkaFHKKS20}]
    \label{def:gkat automaton}
    A \emph{GKAT-automaton} is a pair \((X, \delta)\) consisting of a set \(X\) of \emph{states} and a \emph{transition function} \(\delta \colon \At \times X \to 2 + (\Sigma \times X)\). 
    Here, \(2 = \{\bot, \top\}\), and \(\bot\) and \(\top\) are called \emph{reject} and \emph{accept} respectively.

    Given \(\alpha \in \At\), \(p \in \Sigma\), and a GKAT-automaton \((X, \delta)\) with states \(x,y \in X\), we write 
    \(x \tr{\alpha \mid p} y\) if \(\delta(\alpha, x) = (p, y)\), 
    \(x \Rightarrow \alpha\) if \(\delta(\alpha, x) = \top\), and
    \(x \downarrow \alpha\) if \(\delta(\alpha, x) = \bot\).
    In diagrams, we will not include depictions of \(x \downarrow \alpha\) transitions.

    A guarded string \(\alpha_0 p_1 \cdots p_n \alpha_n\) is \emph{accepted} by a state \(x \in X\) if there is a path 
    \[
        x \tr{\alpha_0 \mid p_1} x_1 \tr{\alpha_1 \mid p_2} \cdots \tr{\alpha_{n-1} \mid  p_n} x_n \Rightarrow \alpha_n,
    \]
    that is, there is a path starting from \(x\) that reads \(\alpha_0 p_1 \cdots p_n\) and then arrives at a state that accepts the atomic test \(\alpha_n\).
    The \emph{language accepted by \(x\)} is given by 
    \[
        \L((X, \delta), x) = \{w \in (\At \cdot \Sigma)^*\cdot \At \mid \text{\(x\) accepts \(w\)}\}.
    \]
\end{definition}

Given a set of word hypotheses \(H\), we use a queue to store a pre-guarded string \(w\) which represents (part of) the sequence of atomic tests and primitive programs that an automaton is reading as it runs.
In particular, the automaton should ``get stuck'' if the state of the queue violates a word hypothesis in \(H\).

\begin{definition}
    \label{def:stuck}
    Let \(H\) be a set of word hypotheses.
    Given \(\alpha \in \At\) and \(w \in (\At \cdot \Sigma)^*\), we say that \emph{\(\alpha\) is locked at \(w\)} if there is an \(n \in \NN\) s.t.\ the guarded string \(\insp(n, w)\alpha\) violates \(H\).
    That is, \(\alpha\) is locked at \(w\) if there is a suffix \(u\) of \(w\) s.t.\ \(u\alpha\) violates some word hypothesis in \(H\).
\end{definition}

We now define the operational model of GKAT with word hypotheses.
The idea is to take a set of word hypotheses \(H\) and a GKAT automaton \((X, \delta)\) and ``explode'' \((X, \delta)\) into the set of all (\emph{program-state}, \emph{queue-state}) pairs. 
This allows us to force a deadlock whenever a hypothesis is violated; more precisely, the transition function only retains those transitions for which the associated atomic test is not locked at the current queue state.

\begin{definition}
    \label{def:bdd queue automata}
    Let \(H\) be a set of word hypotheses and fix the queue length \(l = \max\{\len(w) \mid (w = wc) \in H\}\) to be the maximum length of a word hypothesis in \(H\) (possibly, \(l = \infty\), which occurs when $H$ is infinite).
    Given a GKAT automaton \((X, \delta)\), we define the \emph{inspection queue automaton} by \(\Queue_H(X, \delta) = ((\At \cdot \Sigma)^{\le l} \times X, \delta_H)\), where the transition function 
    \(\delta_H \colon \At \times ((\At \cdot \Sigma)^{\le l} \times X) \to 2 + \Sigma \times ((\At \cdot \Sigma)^{\le l} \times X)\) is defined by
    \begin{equation*}
        \label{eq:queue automaton delta}
        \delta_H(\alpha, (u, x)) 
        = \begin{cases}
            \bot
            &\text{\(\alpha\) is locked at \(u\)} \\
            (p, (\enq_l(\alpha p, u), y))
            &\text{\(\delta(\alpha, x) = (p, y)\) and \(\alpha\) is not locked at \(u\)} \\
            \delta(\alpha, x) 
            &\text{otherwise}
        \end{cases}
    \end{equation*}
    where \(\alpha \in \At\), \(p \in \Sigma\), \(u \in (\At \cdot \Sigma)^{\le l}\), and \(x,y \in X\).
\end{definition}

\noindent
In the third case in the case distinction above, note that if \(\alpha\) is not locked at \(u\) and \(\delta(\alpha, x) \notin \Sigma \times X\), then either \(\delta_H(\alpha, x) = \top\) or \(\delta_H(\alpha, x) = \bot\).

To prove correctness of the construction of \(\Queue_H(X, \delta)\), we use two auxiliary lemmas. 
\cref{lem:path characterization} characterises the inspection queue during runs, and \cref{lem:pathtopath} relates paths in an automaton $(X,\delta)$ to paths in the associated inspection queue automaton \(\Queue_H(X, \delta)\).

\begin{restatable}{lemma}{lemmaPathCharacterization}
    \label{lem:path characterization}
    Let \((X, \delta)\) be a GKAT-automaton, let \(H\) be a set of word hypotheses, and let \(l = \max\{\len(w) \mid (w = wc) \in H\}\). 
    Given \(x \in X\) and $n\geq 0$, consider a path
    \begin{equation*}
        (w_0, x) 
        \tr{\alpha_0 \mid p_1} (w_1, x_1)
        \tr{\alpha_1 \mid p_2} \cdots 
        \tr{\alpha_{n-1} \mid p_n} (w_n, x_n)
    \end{equation*}
    in \(\Queue_H(X, \delta)\).
    Then \(w_n = \insp(l, w_0\alpha_0 p_1 \cdots \alpha_{n-1} p_n)\).
\end{restatable}

\begin{restatable}{lemma}{lemmaPathtoPath}
    \label{lem:pathtopath}
    Let \((X, \delta)\) be a GKAT-automaton, let \(H\) be a set of word hypotheses, and let \(l = \max\{\len(w) \mid (w = wc) \in H\}\).  
    Suppose there is a path
    \[
        x_s \tr{\alpha_0\mid p_1}\cdots \tr{\alpha_{n-1} \mid p_n} x_n \Rightarrow \alpha_{n}
    \] 
  and $\alpha_{j}$ is not locked at $u\alpha_0p_1\cdots \alpha_{j-1} p_j$ for all $0\leq j\leq n$ and a given $u\in (\At\cdot\Sigma)^{\leq l}$. 
    Then the following path exists in \(\Queue_H(X, \delta)\):
    \[
        (u, x_s) 
        \tr{\alpha_0 \mid p_1} 
        \cdots 
      \tr{\alpha_{n-1}\mid p_n} (\insp(l,u\cdot\alpha_0p_1\cdots \alpha_{n-1} p_n), x_n) 
        \Rightarrow \alpha_{n}.
    \]
\end{restatable}

\begin{theorem}[Correctness]
    \label{thm:correctness of queue model}
    Let \(H\) be a set of word hypotheses, \((X, \delta)\) be a GKAT automaton, and write \(U_H = \{w \mid \text{no guarded substring of \(w\) violates \(H\)}\}\).
    Then we have
    \begin{equation}
        \label{eq:queue language}
        \L(\Queue_H(X, \delta), (\varepsilon, x)) = \L((X, \delta), x) \cap U_H
    \end{equation}
\end{theorem}

\begin{proof}
    We show~\eqref{eq:queue language} by arguing inclusion both ways. Let \(l = \max\{\len(w) \mid (w = wc) \in H\}\).

    (\(\subseteq\)) In the forward direction, let \(w = \alpha_0 p_1 \cdots p_n \alpha_n \in \L(\Queue_H(X, \delta), (\varepsilon, x))\). 
    Then by definition, there is a path 
    \begin{equation}
        \label{eq:proof:queue automaton path}
        (\varepsilon, x) 
        \tr{\alpha_0 \mid p_1} (w_1, x_1)
        \tr{\alpha_1 \mid p_2} \cdots 
        \tr{\alpha_{n-1} \mid p_n} (w_n, x_n) \Rightarrow \alpha_n
    \end{equation}
    This implies the existence of a path 
    \begin{equation*}
        x
        \tr{\alpha_0 \mid p_1} x_1
        \tr{\alpha_1 \mid p_2} \cdots 
        \tr{\alpha_{n-1} \mid p_n} x_n \Rightarrow \alpha_n
    \end{equation*}
    in \((X, \delta)\), so \(w \in \L((X, \delta), x)\).
    It now suffices to show that there is no guarded substring of \(w\) that violates \(H\).
    Suppose towards a contradiction that there is a guarded substring \(\alpha_i p_{i+1} \cdots \alpha_{j-1} p_j \alpha_j\) of \(w\) that \emph{does} violate \(H\). 
    Then, since this means there is a word hypothesis in \(H\) that \(\alpha_i p_{i+1} \cdots \alpha_{j-1} p_j \alpha_j\) violates, then \(l \ge j - i\).
    From \cref{lem:path characterization} we obtain that  \(w_j = \insp(l, \alpha_0 p_1 \cdots \alpha_{j-1}p_j)\), and together with $l\geq j-i$ this implies that \(\alpha_i p_{i+1} \cdots \alpha_{j-1} p_j\) is a suffix of \(w_j\). 
    Because \(\alpha_i p_{i+1} \cdots \alpha_{j-1} p_j \alpha_j\) violates some word hypotheses in $H$, we get that $\alpha_j$ is locked at $w_j$. 
    Hence, $\delta_H(\alpha_j,(w_j, x_j))=\bot$.
    This contradicts~\eqref{eq:proof:queue automaton path}, so no such pre-guarded substring can exist. 
    This concludes the forward inclusion.

    (\(\supseteq\)) For the reverse inclusion,
    let \(w \in \L((X, \delta), x)\), assume that no guarded substring of \(w\) violates \(H\), and write \(w = \alpha_0 p_1 \cdots p_n \alpha_n\).
    Again, immediately, we obtain a path 
    in \((X, \delta)\):
    \begin{equation*}
        x
        \tr{\alpha_0 \mid p_1} x_1
        \tr{\alpha_1 \mid p_2} \cdots 
        \tr{\alpha_{n-1} \mid p_n} x_n \Rightarrow \alpha_n.
    \end{equation*}
    Now, for each \(j \le n\) and \(k \in \NN\), 
    \(\insp(k, \alpha_1 p_1 \cdots \alpha_j p_j) \alpha_{j+1}\) is a subword of $w$ and therefore does not violate \(H\) by assumption. By definition, \(\alpha_{j+1}\) is not locked at \(\alpha_1 p_1 \cdots \alpha_j p_j\).
    From \Cref{lem:pathtopath} it follows that there is the path
    \[
        (\varepsilon, x) 
        \tr{\alpha_0 \mid p_1} 
        \cdots 
        \tr{\alpha_{n-1}\mid p_n} (\insp(l,\alpha_1p_1\cdots \alpha_n p_n), x_n)
        \Rightarrow \alpha_{n+1}
    \]
    in \(\Queue_H(X, \delta)\). Therefore, \(w\) is accepted by the state \((\varepsilon, x)\) in \(\Queue_H(X, \delta)\).
\end{proof}

\cref{thm:correctness of queue model} immediately leads us to the following key result behind our decision procedure.

\begin{corollary}
    \label{corr:states equiv up to H}
    Let \(H\) be a set of word hypotheses, let \(e \in \GExp\), and let \((X, \delta)\) be a GKAT-automaton with a state \(x \in X\) such that 
    \(\L(e) = \L((X, \delta), x)\).
    Then \[
        \L_H(e) = \L(\Queue_H(X, \delta), (\varepsilon, x))
    \]
\end{corollary}

\begin{proof}
    Write \(U_H = \{w \mid \text{no guarded substring of \(w\) violates \(H\)}\}\).
    Then, using the definition of $\L_H$, the assumption, and \cref{thm:correctness of queue model} respectively, in that order,
    \begin{align*}
        \L_H(e) = \L(e) \cap U_H  
        = \L((X, \delta_X), x) \cap U_H
        = \L(\Queue_H(X, \delta_X), (\varepsilon, x)). \tag*{\qedhere}
    \end{align*}
\end{proof}

\subparagraph*{Decidability of word hypotheses.}
We are now ready to state our decidability result with a complexity bound that improves on the one offered by the decidability-via-translation construction used in the Hoare hypothesis case (see \cref{thm:hoare decidable,lem:translation size}).
\cref{corr:states equiv up to H} tells us that to decide whether \(\L_H(e) = \L_H(f)\), for \(e,f \in \GExp\), it suffices to check the language equivalence of states in the inspection queue automata that are constructed from the operational (automaton) semantics of \(e\) and \(f\).
This is the crux of the improved decision procedure, which proceeds in three steps as follows.
Given \(e, f \in \GExp\),
\begin{enumerate}
    \item \label{step1} find GKAT-automata \((X, \delta_X)\) and \((Y, \delta_Y)\) with states \(x \in X\) and \(y \in Y\) such that \(\L(e) = \L((X, \delta_X), x)\) and \(\L(f) = \L((Y, \delta_Y), y)\);
    
    \item \label{step2} compute \(\Queue_H(X, \delta_X)\) and \(\Queue_H(Y, \delta_Y)\);
    
    \item \label{step3} return the result of deciding \(\L(\Queue_H(X, \delta_X), (\varepsilon, x)) = \L(\Queue_H(Y, \delta_Y), (\varepsilon, y))\).
\end{enumerate} 
It is clear from \cref{corr:states equiv up to H} that this procedure is correct. 
We are therefore left with analysing the complexity and justifying the claim that it improves on our previous algorithm.
We proceed by analysing each of the three steps individually.

Step~\ref{step1} can be accomplished in a number of different ways~\cite{SmolkaFHKKS20,SchmidKK021}.
In the cited work, it is shown that for any \(e \in \GExp\), a GKAT-automaton \((X, \delta)\) and a state \(x \in X\) can be found in \(\bigO(|e|)\) time such that \(\L(e) = \L((X, \delta), x)\) and \(|X| \in \bigO(|e|)\).
Therefore, if we define \(n = |e| + |f|\), then Step~\ref{step1} can be completed in \(\bigO(n)\) time and the automata \((X, \delta_X)\) and \((Y, \delta_Y)\) that were found have \(\bigO(n)\) many states each.

For Step~\ref{step2}, we obviously need to assume that \(H\) is finite (this assumption will appear in the statement of \cref{thm:word decidable} below).
Let \(l = \max\{\len(w) \mid (w = wc) \in H\}\), \(m = |\At|\), \(k = |\Sigma|\), and write \(N\) for the total number of states in \(\Queue_H(X, \delta_X)\) and \(\Queue_H(Y, \delta_Y)\).
We compute \(N\) explicitly and provide it an asymptotic bound as follows: 
\begin{equation}
    \label{eq:counting inspection queue states}
    \begin{aligned}
        N
        &= |(\At \cdot \Sigma)^{\le l} \times X| + |(\At \cdot \Sigma)^{\le l} \times Y| \\
        &= |(\At \cdot \Sigma)^{\le l}| (|X| + |Y|) \\
        &= \Big(\frac{(mk)^{l+1} - 1}{mk - 1}\Big) (|X| + |Y|) \\
        &\in \bigO((mk)^ln)
    \end{aligned}
\end{equation}
In the last step, we let \(n = |e| + |f|\) and from the discussion above we know \(|X| + |Y| \in \bigO(n)\).

For Step~\ref{step3}, it is observed that inspection queue automata are equivalent to ordinary GKAT automata, and therefore we can use the original decision algorithm for GKAT presented in~\cite{SmolkaFHKKS20}. In~\cite{SmolkaFHKKS20} an algorithm of Hopcroft and Karp~\cite{hopcroft1971linear} is adapted to decide language equivalence of states of GKAT-automata in nearly-linear time (when the number of primitive tests is held constant) with respect to the total number of states in the GKAT-automata involved.
Here, \emph{nearly-linear} means \(\bigO(n \alpha(n))\) time, with \(n\) the size parameter, in which \(\alpha\) denotes the inverse of the Ackermann function.
For us, the number of states involved is \(N\), calculated above in \eqref{eq:counting inspection queue states}, so Step~\ref{step3} is completed in \(\bigO(N \alpha(N))\) many steps. 
Since $N \in \bigO((mk)^l n)$, this is bounded by $\bigO((mk)^l n\cdot \alpha((mk)^l n))$.%
\footnote{%
    There is a subtle point here that is easy to miss.
    It is not trivial to conclude from \(f(n) \in \bigO(g(n))\) that \(\alpha(f(n)) \in \bigO(\alpha(g(n)))\), although this implication is true. 
    \if\arxiv1
        See \cref{app:ackermann} for details.
    \else
    \fi
}

Putting all of this together, we obtain the following decidability result.

\begin{theorem}[Decidability of word hypotheses]
    \label{thm:word decidable}
    Let $e,f\in \GExp$, let \(H\) be a finite set of word hypotheses, and let \(l = \max\{\len(w) \mid (w = wc) \in H\}\) (note that \(l\) is necessarily finite).
    Let \(n = |e| + |f|\), \(m = |\At|\), and \(k = |\Sigma|\). 
    Then it is decidable in \(\bigO((mk)^l n\cdot \alpha((mk)^l n))\) time whether $\L_H(e) = \L_H(f)$.
\end{theorem}

Recall that the complexity bound in \cref{thm:gkat decidability} holds the number of primitive tests constant.
We are also assuming that the primitive tests are held constant in \cref{thm:word decidable} (and in \cref{corr:hoare hypotheses in nl time} below), as well as the number of primitive programs, but we allow \(l\) to vary.
The presence of \(m\) and \(k\) in the complexity bound in \cref{thm:word decidable} are necessary as they form the base of the exponential factor in the runtime.

Returning to the content of earlier sections, observe that a Hoare hypothesis is a word hypothesis of length \(1\).
In the special case where \(H\) is a set of Hoare hypotheses, therefore, \(l = 1\) in \cref{thm:word decidable}.
This immediately leads us to the following corollary of \cref{thm:word decidable}, which improves on the complexity bound given in \cref{lem:translation size}.

\begin{corollary}
    \label{corr:hoare hypotheses in nl time}
    Let \(H\) be a set of Hoare hypotheses and let \(e, f \in \GExp\).
    Then it is decidable in \(\bigO(n \alpha(n))\) time where \(n = |e| + |f|\), whether \(\L_H(e) = \L_H(f)\).
\end{corollary}

The presented decision procedure shows that word hypotheses can be added to GKAT whilst maintaining the nearly linear time complexity of deciding equivalence.

\begin{remark}\label{remark:fresh}
	The decision procedure relies on an automaton construction (\cref{def:bdd queue automata}) which remembers in the queue \emph{all} of the last $l$ actions and atoms, where $l$ is the maximum length of hypotheses in $H$. 
	If $H$ contains only a few (possibly long) hypotheses, then this may not be optimal, and instead one could choose only to store pre-guarded strings that are relevant for some hypothesis in $H$.
	This may reduce the size of the automaton, but it also complicates the definition (and computation) of the transition function.
	In particular, one then needs to remember all pre-guarded strings that are a prefix of some hypothesis in $H$, and in fact it can be that only a \emph{suffix} of the current queue state is a prefix of a hypothesis in $H$.
	We instead chose to present the conceptually simpler construction in (\cref{def:bdd queue automata}).
	
	Another alternative construction may be to reduce word hypotheses to Hoare hypotheses, through the introduction of fresh variables.
	These variables can then uniquely identify a prefix of a hypothesis, and in this way implicitly store these in a queue which then only needs to be of length 1 (at the cost of introducing fresh variables).
	Similarly to the above-mentioned construction, this may work better when there are few hypotheses in $H$, but also complicates the analysis.
	We believe both alternatives are interesting for consideration in future work, especially in the context of an implementation.
\end{remark}

\section{Future work}
We have provided an extension of GKAT with Hoare hypotheses and word hypotheses. 
The introduced semantics for Hoare hypotheses is proved to be complete and decidable in~\cref{sec:hoare hypotheses}. 
In~\cref{sec:words and queues}, we generalised Hoare hypotheses to word hypotheses, and provided a decision procedure that also covers the plain Hoare hypotheses, and is more efficient in that case. 
The decision procedure makes use of an automata-theoretic construction.

An open question is completeness of GKAT with word hypotheses. 
A possible proof strategy would be to reduce word hypotheses to Hoare hypotheses via fresh variables, as discussed in~\Cref{remark:fresh}, and then reuse the completeness result of Hoare hypotheses. 
This would require, among other things, a proof that a derivation using Hoare hypotheses with fresh variables can be translated to a derivation with the equivalent word hypotheses. 
This question is beyond the scope of the current paper but seems fruitful for future research.

A natural question is whether more general hypotheses can be added to GKAT, like in KA. 
This seems very challenging as there is no canonical interpretation of hypotheses like what exists for language models of KA. 
An important difference is that for language models of KA the addition of hypotheses expands the original language semantics, whereas our GKAT semantics with hypotheses shrinks the plain GKAT semantics. 
A semantics for GKAT with hypotheses that expands the original semantics seems unlikely, as it is unclear how to remain deterministic. 
On the other hand, to develop KA with hypotheses based on a semantics that shrinks languages seems a relevant future direction.

A more straightforward question is how Hoare and word hypotheses can be added to bisimulation GKAT~\cite{SchmidKK021}, which is a process-theoretic variation of GKAT that lacks the \emph{late-termination axiom} \(e0 = 0\). 
A related question is whether we can obtain language GKAT as an instance of bisimulation GKAT with the added hypothesis $e0=0$.

\bibliography{refs}

\newpage

\appendix

\section{Additional proofs for \cref{sec:hoare hypotheses}}
\label{app}

In the proofs below we identify $bc$ and $b \wedge c$, implicitly using axiom (S6) from \cref{fig:axioms}.
We begin by observing the following consequence of the \(\GKAT\) axioms (see \cref{fig:axioms}):
\begin{equation}
    \label{eq:restriction}
    \GKAT \vdash b(e +_b f) = be
\end{equation}
We call this \emph{restriction}, and it can be derived as follows.
Let us start with an identity we will call (U4'), which can be shown using (U2), (U4), and (U2), in that order:
\begin{equation}
    e +_b f 
    = f +_{\bar b} e 
    = \bar b f +_{\bar b} e 
    = e +_b \bar b f
    \tag{U4'}
\end{equation}
We also observe that each test \(b \in \BExp\) can really be thought of as syntactic sugar for \(1 +_b 0\), what we will call (B). 
The equation (B) is due to (U1), (\(\BA\)), (U4), (U4'), and (\(\BA\)) in that order:
\begin{equation}
    b 
    = b +_b b 
    = b1 +_b b 
    = 1 +_b b
    = 1 +_b \bar b b
    = 1 +_b 0
    \tag{B}
\end{equation}
Finally, the identity (restriction) is derived as follows:
\begin{align*}
    b(e +_b f)
    &= (1 +_b 0)(e +_b f)           \tag{B} \\
    &= (e +_b f) +_b 0              \tag{U5, S4, S2} \\
    &= e +_b (f +_b 0)              \tag{U3, \(\BA\)} \\
    &= e +_b (1 +_b 0)f             \tag{S2, S4, U5} \\
    &= e +_b bf                     \tag{B} \\
    &= e +_b \bar bbf               \tag{U4'} \\
    &= e +_b 0                      \tag{\(\BA\), S2} \\
    &= (1 +_b 0)e                   \tag{S2, S4, U5} \\
    &= be                           \tag{B} 
\end{align*}

\lemmaInvarianceEquivToHoisting*

\begin{proof}
    For the right-to-left direction, we have the following derivation of \(p(e +_b f) = pe +_b pf\):
    \begin{align*}
        \GKAT + \mathsf{Inv}_{p,b} \vdash 
        ~~p(e +_b f)  
        &= p (e +_b f) +_b p(e +_b f) \tag{U1}\\
        &= bp (e +_b f) +_b \bar bp(e +_b f) \tag{U4, U4'} \\
        &= bpb (e +_b f) +_b \bar bp\bar b(e +_b f) \tag{\(\mathsf{Inv}_{p,b}\)}\\
        &= bpbe +_b \bar bp\bar bf \tag{restriction, U2} \\
        &= bpe +_b \bar bpf \tag{\(\mathsf{Inv}_{p,b}\)}\\
        &= pe +_b pf \tag{U4, U4'}
    \end{align*}
    In other words, \(\mathsf{Inv}_{p,b}\) implies that \(p\) can be hoisted from any conditional \(e +_b f\).

    Conversely, the hoisting hypothesis \(\mathsf{Hst}_{p,b}\) implies the invariant test hypothesis \(\mathsf{Inv}_{p,b}\). 
    To see this, choose \(e = 1\) and \(f = 0\) and observe that
    \begin{align*}
        \GKAT + \mathsf{Hst}_{p,b}
        \vdash~~
        bpb 
        &= bp(1 +_b 0) \tag{B}\\
        &= b(p1 +_b p0) \tag{\(\mathsf{Hst}_{p,b}\)}\\
        &= bp1 \tag{restriction} \\
        &= bp 
    \end{align*}
    The same argument with \(e = 0\) and \(f = 1\) shows that the hoisting principle for \(p\) over \(b\) implies \(\bar b p = \bar b p \bar b\).
\end{proof}

\lemmaCommuteEquivToInvariance*

\begin{proof}
In the forward direction,
\begin{align*}
    \GKAT + \mathsf{Inv}_{p,b} \vdash 
    pb
    &= pb +_b pb \tag{U1} \\
    &= bpb +_b \bar b pb \tag{U4, U4'}\\
    &= bp +_b \bar b p \bar b b \tag{\(\mathsf{Inv}_{p,b}\)} \\
    &= bp +_b \bar b p 0 \tag{BA}\\
    &= bp +_b 0 \tag{S3} \\
    &= bp \tag{U4, S2, S4, U5, B}
\end{align*}
The reverse direction is much simpler:%
\footnote{%
    Interestingly, it appears that the use of (S3) is necessary in the forward direction, but the reverse direction does not require (S3).
}
\begin{align*}
    \GKAT + \mathsf{Com}_{p,b} \vdash 
    bpb
    &= bbp \tag{\(bp = pb\)} \\
    &= bp \tag{\(\BA\)}  
\end{align*}

Similar derivations can be given for $\bar b p =\bar b p \bar b$ and $\bar b p = p \bar b$.
\end{proof}

\lemmaLHHom*

\begin{proof}
	Note that $\L_H(b) = \L(b)$ follows easily from the definition of $\L_H$, giving the first case.
	For the remaining cases, we first define
    $$ X = \{ \alpha_1p_1\cdots\alpha_np_n\alpha_{n+1} \mid (b p_i = b p_i c\in H\wedge \alpha_i\leq b) \Rightarrow \alpha_{i+1}\leq c\} $$ 
	and observe that for any \(e \in \GExp\), \(\L_H(e) = \L(e) \cap X\). 
	
    We will need the following identity: for any guarded languages \(L_1,L_2\), we have
    \begin{align*}
        (L_1 \diamond L_2) \cap X &= (L_1 \cap X) \diamond (L_2 \cap X)
        \tag{\(\dagger\)}
    \end{align*}
    The identity (\(\dagger\)) above is derived from the fact that we have a guarded string $u\alpha v\in X$ if and only if we have guarded strings $u\alpha\in X$ and $\alpha v\in X$.
    In particular, one should note that \(\At \subseteq X\), so from (\(\dagger\)) we also obtain \((b \diamond L) \cap X = b \diamond (L \cap X)\) and \((L \diamond b) \cap X = (L \cap X) \diamond b\) for any guarded language \(L\) and \(b \in \BExp\).

    In the guarded union case, we have
    \begin{align*}
        \L_H(e +_b f)
        &= (b\diamond \L(e) \cup \bar b\diamond \L(f))\cap X\\
        &= ((b\diamond \L(e))\cap X) \cup ((\bar b\diamond \L(f))\cap X) \\
        &= (b\diamond (\L(e)\cap X) )\cup (\bar b\diamond (\L(f)\cap X)) \tag{\(\dagger\)} \\
		&= (b\diamond \L_H(e)\cup (\bar b\diamond \L_H(f))
    \end{align*}
    For concatenation, we have
    \begin{align*}
        \L_H(e \cdot f )&= (\L(e) \diamond \L(f))\cap X \\
        &= (\L(e) \cap X)\diamond (\L(f)\cap X) \tag{$\dagger$}\\
        &= \L_H(e) \diamond \L_H(f) \tag{def.~of \(\L_H\)}
    \end{align*}
    For the star case, we first derive the following for all $n\geq 0$:
    \begin{equation}\label{eq:starx}
        (b\diamond L)^{\langle n \rangle }\cap X = (b\diamond (L\cap X))^{\langle n \rangle }
    \end{equation}
    We proceed by induction on $n$. 
    The base case is trivial as $\At\subseteq X$. 
    For $n+1$ we derive 
    \begin{align*}
        (b\diamond L)^{\langle n+1 \rangle} \cap X &= ((b\diamond L)\diamond (b\diamond L)^{\langle n \rangle})\cap X\\
        &= ((b\diamond L) \cap X)\diamond ((b\diamond L)^{\langle n \rangle}\cap X) \tag{$\dagger$}\\
        &= (b\diamond (L\cap X))\diamond ((b\diamond L)^{\langle n \rangle}\cap X) \tag{$\dagger$} \\
        &=(b\diamond (L\cap X))\diamond (b\diamond (L\cap X))^{\langle n \rangle}\tag{I.H.}\\
        &= (b\diamond (L\cap X))^{\langle n+1 \rangle}
    \end{align*}
    Returning to the \((-)^{(b)}\) case, recall that for a guarded language \(L\), \(L^\star = \bigcup_{n\ge 0} L^{\langle n \rangle}\).
    We have
    \begin{align*}
        \L_H(e^{(b)})&= ((b \diamond \L(e))^\star \diamond \bar b)\cap X\\
        &= ((b \diamond \L(e))^\star \cap X)\diamond \bar b \tag{\(\dagger\)} \\
        &= \left(\left( \bigcup_{n \geq 0}  \left(b \diamond \L(e)\right)^{\langle n \rangle }\right)\cap X\right) \diamond \bar b \\
        &= \left(\bigcup_{n \geq 0}  \left(\left(b \diamond \L(e)\right)^{\langle n \rangle } \cap X\right)\right) \diamond \bar b \\
        &= \left(\bigcup_{n \geq 0}  \left(b \diamond (\L(e)\cap X)\right)^{\langle n \rangle } \right)\diamond \bar b \tag*{\eqref{eq:starx}}\\
        &= (b \diamond (\L(e)\cap X))^\star \diamond \bar b \\
		&= (b \diamond \L_H(e))^\star \diamond \bar b \qedhere
        \end{align*}
\end{proof}

\theoremSoundness*

\begin{proof}
    The derivation relation is the smallest congruence relation that contains the axioms of $\GKAT$ and the equalities in $H$. 
    By showing that the language equivalence relation is a congruence that contains the axioms of $\GKAT$ and the equalities in $H$, it follows that \(\GKAT +H \vdash e = f\) implies $\L_H(e) = \L_H(f)$. 
    If  \( e = f\) is an axiom of $\GKAT$, we obtain from $\GKAT$ soundness that $\L(e)=\L(f)$. 
    For $e=f\in H$, we know that $e=bp$ and $f=bpc$, and the result follows immediately. 
	
	It remains to show that language equivalence up to $H$ is a congruence.
	Reflexivity, symmetry and transitivity are trivial. 
	It remains to show closure under each of the operators. 
	To this end, let $\L_H(e_0) = \L_H(f_0)$ and $\L_H(e_1) = \L_H(f_1)$. 
	Then 
	$$
	\L_H(e_0 +_b f_0) = b\diamond \L_H(e_0) \cup \bar b\diamond \L_H(f_0)
	= b\diamond \L_H(e_1) \cup \bar b\diamond \L_H(f_1)
	= \L_H(e_1 +_b f_1)
	$$
	using \cref{lem:lh-hom}. 
	The cases $\L_H(e_0 \diamond f_0) =  \L_H(e_1 \diamond f_1)$ and $\L_H(e_0^{(b)}) = \L_H(f_0^{(b)})$ follow similarly from \cref{lem:lh-hom}.
\end{proof}

\translationsize*

\begin{proof}
    We proceed by induction on $e$. If $e=b\in\BExp$, the result trivially holds because $\transl{H}(b)=b$. 

    For $e=p\in\Sigma$, we proceed to prove that $|\transl{H}(p)|=3(2^{n_p^H})-2$ by induction on the number of hypotheses in $H$ that involve $p$. If there exists no hypothesis in $H$ that involves $p$, the result trivially holds. 
    Otherwise pick any $h\in H$ of the form $bp=bpc$ for some $b,c\in\BExp$.
    As the order of reductions in $\transl{H}$ does not matter (\cref{lem:hypcommute}, and it is also easy to see it does not affect the size), we assume the reduction $\transl{h}$ is performed `last', i.e. $\transl{H}(e)=\transl{ h}(\transl{ H\setminus h}(e))$. To alleviate notation, we denote $n_p^H$ with $n$ and $n_p^{H\setminus h}$ with $m$. Observe that $n=m+1$.
    Our induction hypothesis states that $|\transl{H\setminus h}(p)|=3(2^m)-2$. We derive 
    \begin{align*}
    |\transl{H}(p)| &= |\transl{h}(\transl{H\setminus h}(p))|\\
    &=|\transl{H\setminus h}(\transl{h}(p))| \\
    &= |\transl{H\setminus h}(pc+_b p)| \\
    &= | \transl{H\setminus h}(p)c+_b \transl{H\setminus h}(p)|\\
    &=2|T_{H\setminus h}(p)|+2 \\
    &= 2(3(2^m)-2)+2 \tag{Induction Hypothesis}\\
    &=3(2^{m+1})-2 = 3(2^n)-2 \tag{$2(2^m)=2^{m+1}$, $n=m+1$}
    \end{align*}

    For $e+_b f$ we abbreviate $n_p^H$ with $n_p$ and derive 
    \begin{align*}
    &|\transl{H}(e+_b f)| \\
    &= |\transl{H}(e)+_b \transl{H}(f)|\\
    &= |\transl{H}(e)|+ 1 +| \transl{H}(f)|\\
    &=  \sum_{p\in e}(3(2^{n_p})-3)+|e| +1 + \sum_{p\in f}(3(2^{n_p})-3)+|f|\tag{Induction Hypothesis}\\
    &=\sum_{p\in e+_b f}(3(2^{n_p^H})-3)+|e+_b f|
    \end{align*}
    The other inductive cases follow similarly. 
\end{proof}

\section{Additional proofs for \cref{sec:words and queues}}
\label{app2}

\lemmaPathCharacterization*

\begin{proof}
    We proceed by induction on $n$. For $n=0$, $\insp(l,w_0)=w_0$ follows because $w_0\in(\At\cdot \Sigma)^{\leq l}$.

    For $n+1$, we derive 
    \begin{align*}
        w_{n+1} 
        &= \enq_l(\alpha_{n+1} p_{n+1}, w_{n}) \tag{def.~of \(\delta_H\)} \\
        &= \insp(l, w_{n}\alpha_{n+1} p_{n+1}) \tag{def.~of \(\enq_l\)} \\
        &= \insp(l, w_0\alpha_1 p_1 \cdots \alpha_{n+1} p_{n+1} ) \tag{ind.~hyp.}
    \end{align*}
    as desired. 
    This concludes the proof.
\end{proof}

\lemmaPathtoPath*

\begin{proof}
    We proceed by induction on $n$. 
    If $n=0$, then $\alpha_0$ is not locked at $u$ by assumption. 
    Hence, $\delta_H(\alpha_0,(u,x_s))=\delta(\alpha_0,x)=\top$, and we can conclude. 
    
    For the induction step, let \(n > 0\) and consider the path 
    \[
      x_1 \tr{\alpha_1\mid p_2} \cdots \tr{\alpha_{n-1} \mid p_n} x_n \Rightarrow \alpha_n
    \]
    starting at \(\alpha_1\).
    We relabel by setting $\beta_i=\alpha_{i+1}$ for $i\in \{0,\dots n-1\}$ and $q_i=p_{i+1}$ for $i\in \{1,\dots n-1\}$.
    This gives the path
    \[
        x_1 \tr{\beta_0\mid q_1} \cdots \tr{\beta_{n-2} \mid q_{n-1}} x_n \Rightarrow \beta_{n-1}
    \] 
    Now, for all $n\in \NN$, the guarded string
    $$
      \insp(n,u \alpha_0 p_1 \alpha_1 p_2\cdots \alpha_{j-1}p_{j})\cdot \alpha_{j}
    $$ 
    does not violate $H$ for all $0\leq j\leq n$ by assumption.
    Consequently, for all $n\in \NN$, the string 
    $$
      \insp(n,u \alpha_0 p_1 \beta_0 q_1\cdots \beta_{k-2}q_{k-1})\cdot \beta_{k-1}
    $$ 
    does not violate $H$ for any $1\leq k\leq n$. 
    Hence,  $\beta_{k-1}$ is not locked at $\insp(l,u \alpha_0 p_1 \beta_0 q_1\cdots \beta_{k-1}q_{k-1})$ for any $1\leq k \leq n-1$. 
    Then we can apply the induction hypothesis to obtain a path,
    \[
        (u \alpha_0 p_1, x_1) 
        \tr{\beta_0 \mid q_1} 
        \cdots 
      \tr{\beta_{n-2} \mid q_{n-1}} (\insp(l,u\cdot\alpha_0p_1\beta_0 q_1\cdots \beta_{n-2} q_{n-1}), x_n) 
        \Rightarrow \beta_{n-1}
    \]
    If we reverse our labelling, we obtain the path
    \[
        (u \alpha_0 p_1, x_1) 
        \tr{\alpha_1 \mid p_2} 
        \cdots 
        \tr{\alpha_{n-1} \mid p_{n}} (\insp(l,u\cdot\alpha_0p_1\alpha_1 p_2\cdots \alpha_{n-1} p_n), x_n) 
        \Rightarrow \alpha_n
    \]
    As $\delta(\alpha_0, x_s)=(p_1,x_1)$ and $\alpha_0$ is not locked at $u$ by assumption, we get that $\delta_H(\alpha_0,u,x_s)=(p_1,\insp(l,u\alpha_0 p_1), x_1)$, and we can conclude.
\end{proof}

\if\arxiv1
\subsection*{Proof of the complexity bound in \cref{thm:word decidable}}

\label{app:ackermann}

It is a subtle point, but the complexity bound in \cref{thm:word decidable} requires a nontrivial property of the inverse Ackermann function.
The point where this property is used was in the step where we deduced from \(N \in \bigO((mk)^ln)\) that \(N \alpha(N) \in \bigO((mk)^l n\alpha((mk)^ln))\).
This would of course follow if \(f(n) \in \bigO(g(n))\) implies \(\alpha(f(n)) \in \bigO(\alpha(g(n)))\) for any \(f,g\).
But this is not true for just any function in place of \(\alpha\)! 
This \emph{complexity preservation property} of \(\alpha\) is needed to ensure the validity of this step.
This complexity preservation property of \(\alpha\), stated below (\cref{thm:ackermann preservation}), is likely to be known, but we provide a proof anyway for the convenience of the reader.

To state the property, we need to fix a defintion of the Ackermann function and its inverse. 
We choose a standard definition, found in~\cite{CLRS}.
For each \(k \in \NN\), the \emph{\(k\)-th Ackermann function} \(A_k \colon \NN \to \NN\) is defined inductively by 
\[
    A_k(m) = \begin{cases}
        m + 1 &\text{if }k = 0 \\
        A_{k-1}^{m+1}(m) &\text{if }k >0 
    \end{cases}
\]
where we are using the notation \(f^m(x)\) to denote the \(m\)-fold application of some function \(f\) to its argument \(x\), where \(m\) is any nonnegative integer.
The \emph{inverse Ackermann function} \(\alpha \colon \NN \to \NN\) is defined to be 
\[
    \alpha(n) = \min\{k \in \NN \mid A_k(1) \ge n\}
\]
We can now state the property of \(\alpha\) needed in \cref{corr:hoare hypotheses in nl time}.
Write \(\vec n = (n_0, n_1, \dots, n_{l-1})\) for a vector of natural numbers.

\begin{theorem}
    \label{thm:ackermann preservation}
    Let \(f, g \colon \NN^l \to \NN\) be positive integer valued functions in some fixed number of arguments, and assume they are monotone in each argument. 
    Let \(\alpha \colon \NN \to \NN\) be the inverse Ackermann function. 
    Then, if \(f(\vec n) \in \mathcal O(g(\vec n))\), we also have \(\alpha(f(\vec n)) \in \mathcal O(\alpha(g(\vec n)))\).
\end{theorem}

The proof proceeds in several stages. 
We need the following standard properties of the Ackermann functions: for any \(k,m \in \NN\), 
\begin{enumerate}
    \item \label{ack prop 1} \(A_{k+2}(m) \ge A_k(2m)\)
    \item \label{ack prop 2} if \(k_1 \ge k_2\), then \(A_{k_1}(m) \ge A_{k_2}(m)\)
    \item \label{ack prop 3} if \(m_1 \ge m_2\), then \(A_k(m_1) \ge A_k(m_2)\)
    \item \label{ack prop 4} for any \(k \in \NN\), \(\log(m) \in \bigO(A_k(m))\)
\end{enumerate}
From these, we can derive the following inequality. 

\begin{lemma}
    \label{lem:ackermann help}
    For any \(k \in \NN\), there is a number \(d\) such that for any \(l \ge d\), 
    \[A_{k + 2l}(1) \ge (l/d) \cdot A_k(1)\]
\end{lemma}

\begin{proof}
    The proof is split into two inequalities: first, we show that \(A_{k + 2l}(1) \ge A_k(2^l)\) for all \(k \in \NN\) and \(l \ge 2\).
    Then we are going to find \(d\) such that \(d \cdot A_k(2^l) \ge c \cdot A_k(1)\) for all \(l \ge d\).

    For the first statement, we show something slightly stronger, that for \(m \in \NN\), we have \(A_{k + 2l}(m) \ge A_k(m2^l)\).
    We proceed by induction on \(l\).
    In the base case, \(l = 0\), so 
    \[
        A_{k + 2l}(m) = A_{k}(m) = A_k(m2^l)
    \]
    because \(2l = 0\) and \(2^l = 1\).
    In the induction step, we have 
    \[
        A_{k+2(l+1)}(m)
        = A_{k+2l+2}(m)
        \ge A_{k+2l}(2m)
        \ge A_{k}(2m2^l)
        = A_{k}(m2^{l+1})
    \]
    The first inequality is property~\ref{ack prop 1}.
    The second inequality is the induction hypothesis.
    Taking \(m = 1\), we obtain the desired statement.

    Let us now show the second statement.
    For any fixed \(k \in \NN\) and varying \(j\), we immediately obtain \(\log(j) \cdot A_k(1) \in \bigO(A_k(j))\) from property~\ref{ack prop 4}.
    This means that there are constants \(m',m'' \in \NN\) such that \(\log(j) \cdot A_k(1) \le m' \cdot A_k(j)\) for all \(j \ge m''\).
    If we take \(m = \max\{m',m''\}\), then this becomes the statement that \(\log(j) \cdot A_k(1) \le m \cdot A_k(j)\) for all \(j \ge m\).
    If we substitute \(j = 2^l\) and \(d = 2^m\), then we obtain the desired \(l \cdot A_k(1) \le d \cdot A_k(2^l)\) for all \(l \ge d\).

    Putting these two together, for any \(k\), we can find a \(d \in \NN\) such that for any \(l \ge d\), 
    \[
        A_{k+2l}(1) \ge A_k(2^l) \ge (l/d) A_k(1)
    \]
    as desired.    
\end{proof}

We can now move on to the key lemma.

\begin{lemma}
    \label{lem:ack plus 1}
    For any integer \(c \ge 2\), there is a \(t \in \NN\) such that \(\alpha(c \cdot n) \le \alpha(n) + t\) for all \(n\).
\end{lemma}

\begin{proof}
    Unravelling definitions, we are trying to find a \(t\) such that 
    \[
        \alpha(c \cdot n) 
        = \min\{k \mid A_k(1) \ge c \cdot n\}
        \le \min\{k \mid A_k(1) \ge n\} + t
        = \alpha(n) + t
    \]
    for all \(n\), i.e., \(t\) does not depend on \(n\).
    Subtracting \(t\) from both sides, this inequality is equivalent to 
    \begin{equation}
        \label{eq:inexack}
        \min\{k \mid A_k(1) \ge c \cdot n\} - t
        \le \min\{k \mid A_k(1) \ge n\}
    \end{equation}
    By substitution, the left hand side of \eqref{eq:inexack} is equivalent to 
    \[
        \min\{k - t \mid A_{k}(1) \ge c \cdot n\}
        = \min\{k\mid A_{k+t}(1) \ge c \cdot n\}
    \]
    This tells us that, given that we have chosen \(t\) correctly, it suffices to show that for any \(k\) with \(A_k(1) \ge n\), there is a \(k' \le k\) such that \(A_{k' + t}(1) \ge c \cdot n\).

    Fix \(k \in \NN\) such that \(A_k(1) \ge n\). 
    If we choose \(k' = k\), then we need to find a \(t\) such that \(A_{k + t}(1) \ge c \cdot n\).
    By \cref{lem:ackermann help}, there is a number \(d\) such that for any \(l \ge d\), \(A_{k + 2l}(1) \ge (l/d) \cdot A_k(1)\).
    Let us choose \(t = 2cd\).
    Then \(c \ge 2\) means that \(cd \ge d\), and we therefore have
    \[
        A_{k'+t}(1)
        = A_{k+2cd}(1)
        \ge (cd/d) \cdot A_k(1)
        = c \cdot A_k(1)
        \ge c \cdot n
    \]
    Since \(t\) does not depend on \(n\), this is true for all \(n \in \NN\), as desired.
\end{proof}

We are now ready to prove \cref{thm:ackermann preservation}.

\begin{proof}
    [Proof of \cref{thm:ackermann preservation}.]
    Since \(f(\vec n) \in \bigO(g(\vec n))\), there is a constant \(c \ge 2\) and a number \(n_0\) such that for any \(\vec n\) with \(|\vec n| > n_0\) we have \(f(\vec n) \le c\cdot g(\vec n)\).
    Fix such a constant \(c\), and let \(t\) be the constant found in \cref{lem:ack plus 1} (which depends on \(c\)) that ensures \(\alpha(c\cdot n) \ge \alpha(n) + t\) for all \(n\).
    Since \(\alpha\) is monotone, we therefore have 
    \[
        \alpha(f(\vec n))
        \le \alpha(c \cdot g(\vec n))
        \le \alpha(g(\vec n)) + t
    \]
    It follows that \(\alpha(f(\vec n)) \in \bigO(\alpha(g(\vec n)))\).
\end{proof}

\fi 

\end{document}